\documentclass[12pt,onecolumn,Letter,draftclsnofoot]{IEEEtran}
\usepackage{amssymb}
\usepackage{amsfonts}
\usepackage{amsthm}
\usepackage{amsmath}
\usepackage{mdwlist}
\usepackage{url}
\usepackage{afterpage} %to flush the floats after the page
\usepackage {graphicx}
\usepackage{multirow}
\usepackage{cite}
\usepackage{bigstrut} %to put more space for cells inside a table
\usepackage{subcaption}
\usepackage{algorithm}
\usepackage{algpseudocode,epstopdf,epsfig}
\usepackage{xcolor}
\usepackage{tikz}
\usepackage{pgfplots}
\title{Maximal Packing with Interference Constraints}
\author{Rakshith Jagannath, Radha Krishna Ganti and Neelesh S Upadhye 
        \thanks{Raksith Jagannath, Radha Krishna Ganti and Neelesh S Upadhye are with Indian Institute of Technology, Madras, India; e-mail:\{rejarr@gmail.com, rganti@ee.iitm.ac.in, neelesh@iitm.ac.in\}}}
\def\d{\mathbf{d}}
\def\e{\mathbf{e}}

\def\l{\ell}
\def\n{\mathbf{n}}
\def\o{\mathbf{1}}
\def\r{\mathbf{r}}
\def\v{\mathbf{v}}
\def\w{\mathbf{w}}
\def\x{\mathbf{x}}

\def\u{\mathbf{u}}
\def\A{\mathbf{A}}
\def\D{\mathbf{D}}
\def\F{\mathbf{F}}

\def\H{\mathbf{H}}
\def\I{\mathbf{I}}

\def\Q{\mathbf{Q}}

\def\R{\mathbf{R}}
\def\diag{\mathbb{D}}
\def\mean{\mathbb{E}}
\def\E{\mathbb{E}}

\def\el{\mathbb{L}}
\def\pr{\mathbb{P}}

\def\sgn{\mathbb{S}}
\def\tr{\mathbb{T}}
\def\ie{\emph{i.e.}}

\def\etc{\emph{etc}}

\def\N{\mathcal{N}}
\newtheorem{thm}{Theorem}

\newcommand{\NrmOne}[1]{\|{#1}\|_1}
\newcommand{\NrmTwo}[1]{\|{#1}\|_2}
\newcommand{\NrmZ}[1]{\|{#1}\|_0}
\DeclareMathOperator*{\argmax}{arg\,max}

\def\rj[#1]{{\color{blue}\sffamily\small\em $\Rightarrow$ #1$\Leftarrow$}}
\begin{document}
\maketitle
%-----------------------------
%Abstract Beginning
%-----------------------------
\begin{abstract}
In this work, we study the problem of scheduling a maximal set of transmitters subjected to an interference constraint across all the nodes. Given a set of nodes, the problem reduces to finding the maximum cardinality of a subset of nodes that can concurrently transmit without violating interference constraints. The resulting packing problem is a binary optimization problem and is  NP hard. We propose a semi-definite relaxation (SDR) for this problem and provide bounds on the relaxation.
\end{abstract}
%-----------------------------
%Abstract Ending
%-----------------------------
%---------------------------------------------------------------
% Introduction Start
%----------------------------------------------------------------
\section{Introduction}
Interference is a major impediment in the current wireless networks, particularly in ad-hoc and wireless sensor networks. In these networks, interference is primarily managed through scheduling wherein the transmitting nodes are carefully chosen to avoid interference at the active links, while simultaneously maximizing the spatial reuse. The maximum number of nodes that can be spatially scheduled with a network interference constraint is an important metric that quantifies the performance of the scheduler. However, to the best of our knowledge, even this simple metric is difficult to be computed for general network topologies.

In this paper, we focus on the problem of computing the cardinality of the largest subset of nodes that can be scheduled from a given set of nodes with a constraint on the interference across the network. We assume an arbitrary network topology and model the spatial interference pattern through a path-loss function. The cardinality of the maximal set can be obtained by solving a binary problem, which is NP-hard. Hence, we obtain bounds on the maximal node packing with interference constraints through a semi-definite relaxation (SDR) of the original binary problem by using Shor's technique \cite{Ben-Tal,sdr-survey}.

In \cite{gupta2000capacity}, the maximal packing problem has been studied with the protocol model for interference with results from random geometric graphs. It has been shown that the maximal density of scheduled nodes scales as $O(1/\sqrt{N})$, where $N$ is the total number of nodes. In \cite{boyd,SS}, a related problem of sensor selection, \ie, selecting $K$ out of $N$ sensors that minimize the error in estimating network parameters is studied and solutions are proposed using several frameworks such as convex optimization, hypothesis testing, experiment design, compressed sensing  and sparse signal recovery \etc. Another related problem is the signal-to-interference-and-noise ratio (SINR) maximization problem wherein the SINR at each node is maximized \cite{SIR,SIR-1,SIR-2} using techniques from semi-definite programming and graph theory. However, the above methods assume that the maximum number of nodes is fixed and the interference amongst the selected nodes is optimised. In the current work, a more fundamental question of finding the maximum number of nodes  for a given  interference constraint across the network is explored.

{\em Notation}: In this paper, we use bold lower case letters to represent vectors and bold upper case letters to represent matrices. For a given matrix (vector) $\A$, $\A^T$ denotes the regular transpose. For a vector $\x$,  $\NrmZ{\x}$, $\NrmOne{\x}$, $\NrmTwo{\x}$ denote the $l_0$ pseudo norm which is equal to the number of non-zero elements in $\x$, $l_1$ and $l_2$ norms respectively.  $\mean$ denotes the expectation operation while $\pr$ denotes probability.  For a matrix $\A$ and a vector $\u$, $\diag(\A)$ denotes a vector of diagonal entries of $\A$, $\tr(\A)$ denotes the trace of $\A$, $\sgn(\mathbf{u})$ denotes the sign of elements of $\mathbf{u}$, $\sgn(u)=+1$, if $u>0$, else $\sgn(u)=-1$. $\arcsin\A$ denotes sine-inverse of each element of $\A$ (see \cite{Ben-Tal}) and $\A\succeq\mathbf{0}$ implies $\x^{T}\A\x\geq0$ for all $\x\neq0$.
%---------------------------------------------------------------
% Introduction End 
%----------------------------------------------------------------
%\newpage
%---------------------------------------------------------------
% Signal Model Start
%----------------------------------------------------------------
\section{Signal Model}\label{model}
We consider $N$ nodes located at $\{\v_1, \v_2, \hdots, \v_N\} \subset \mathbb{R}^2 $.  The Euclidean distance between node $i$ and node $j$ is denoted by $r_{ij}$. In this paper, we neglect thermal noise and assume free space channel between nodes\footnote{Fading is neglected so as to simplify the notation and can be introduced without any modifications to the results.}. The path loss function is denoted by the function $\l(x):\mathbb{R}^2 \to [0,\infty)$. A commonly used path loss function is $\l(\v) =\|\v\|_{2}^{-\beta}$, where $\beta >2$ is the path loss exponent. Let $x_{j}\in \{0,1\}$ denote an indicator variable which equals one if node $j$ is active and zero otherwise.  Assuming unit power transmission,  the interference power at a node $i$ due to  other transmitting nodes  is
\begin{align}
w_i = \sum_{j = 1}^{N}  \l(r_{ij})x_{j}, \quad  i = 1,2,\ldots,N.
\end{align}
 Let \[\d_{i} = \begin{bmatrix}\l(r_{i1}),\ldots,0,\l(r_{ii+1}),\ldots,\l(r_{iN})\end{bmatrix}^{T},\] and
$\x = \begin{bmatrix}x_1,x_2,\ldots,x_N\end{bmatrix}^{T}$. Then the interference (or received signal power) at node $i$  is
\begin{equation}
w_i = \d_{i}^{T}\x.
\end{equation}
Let $\w=\begin{bmatrix}w_1,w_2,\ldots,w_N\end{bmatrix}^{T}$. Then $\w = \mathbf{Dx}$,
where  $\D = \begin{bmatrix}\d_{1},\d_{2},\ldots,\d_{N}\end{bmatrix}^{T}$ is the distance matrix whose elements are non-negative.

Using the above notation, the maximum number of nodes in the network while limiting the interference powers across all the nodes to be less than $\epsilon$, is given by
\begin{equation}
\sigma = \max_{\x\in{\{0,1\}}^{N}} \NrmZ{\x}\hspace{2mm}\mathrm{s.t.}\hspace{2mm}\NrmTwo{\w}^2\leq\epsilon.
\end{equation}
In the above optimization problem, we consider the two norm of the interference across all the nodes for analytical tractability. However, since $\|\w\|_2 \geq \|\w\|_\infty$, the constraint $\NrmTwo{\w}^2\leq\epsilon$ implies a bound on the interference at individual nodes.

Since, $\x\in{\{0,1\}}^{N}$, we have $\NrmZ{\x} = \NrmOne{\x} = \NrmTwo{\x}^2$. Also, $\w = \mathbf{Dx}$, and hence the above optimization problem can be rewritten as
\begin{equation}
\sigma = \max_{\x\in{\{0,1\}}^{N}} \NrmTwo{\x}^2\hspace{2mm}\mathrm{s.t.}\hspace{2mm}\x^{T}\F\x\leq\epsilon,
\label{mainoptone}
\end{equation}
where  $\F = \D^{T}\D$ is a symmetric positive semi-definite matrix. The properties of $\D$ and $\F$ are discussed in detail in \cite{distance1}. The optimization problem (\ref{mainoptone}) is an NP hard \cite{bestsdp}, discrete optimization problem and there are no closed form analytical solutions to the above problem.

Observe that $\epsilon$ is a network wide interference constraint and hence for a uniform network, it is easy to observe that  $\epsilon$ should scale with $N$ if a significant subset of nodes have to be activated.  Also note that $\sigma\leq N$ is a trivial bound, since the number of active nodes is always bounded above by the total number of nodes.

In this paper, we first obtain a SDR of  (\ref{mainoptone}).  Since SDR is a relaxation of (\ref{mainoptone}), the optimum value of the SDR ($\rho$) is related to $\sigma$ as $\rho\leq\sigma$.  This relaxation is then used to  propose a randomization algorithm (rounding) to obtain a $\x$ for (\ref{mainoptone}) also called as rounding in literature \cite{GWbook}. This rounding technique is then used to  obtain bounds of the type $\rho\leq\theta\sigma$ for some constant $\theta>1$.

%---------------------------------------------------------------
% SDR Start
%----------------------------------------------------------------
\section{Semi-definite Relaxation}
%-------------------------------------------------------------------

The binary problem in \eqref{mainoptone} is first converted into a $\{-1,+1\}$ problem using the transformation $\v = 2\x-\o$. Hence the equivalent problem is 
\begin{eqnarray}
\label{eq:modified}
\sigma = \max_{\v}\frac{1}{4}\big(\v^{T}\v+\o^{T}\v+\v^{T}\o+ \o^{T}\o\big),\\
\mathrm{s.t.}\hspace{2mm} \v^{T}\F\v+\o^{T}\F\v+ \v^{T}\F\o+\o^{T}\F\o
\leq 4\epsilon, \nonumber\\
v_i^2 = 1, i = 1,2,\ldots,N. \nonumber
\end{eqnarray}
The above optimization problem can be rewritten as
\begin{align}
\sigma = \max_{\v}\frac{1}{4} \begin{bmatrix}\v^{T}& 1\end{bmatrix}\Q\begin{bmatrix}\v\\ 1 \end{bmatrix},\label{optgood}\\
\mathrm{s.t.}\hspace{2mm}
\begin{bmatrix}\v^{T}& 1\end{bmatrix}\R\begin{bmatrix}\v\\ 1 \end{bmatrix}\nonumber\leq 4,\\
\v^{T}\big(\e_i\e_i^{T}\big)\v = 1, \quad i = 1,2,\ldots,N,
\end{align}
where $\e_i$ is the  column $i$  of the  identity matrix and $\Q$ and $\R$ are given by
\begin{align}
\Q =
\begin{bmatrix} \I&&\o\\ \o^{T}&&\o^{T}\o\end{bmatrix}, &&
\R =
\frac{1}{\epsilon}\begin{bmatrix} \F&&\mathbf{F1}\\ \o^{T}\F&&\o^{T}\F\o\end{bmatrix}.
\end{align}
Observe that $\begin{bmatrix}\v^{T}& 1\end{bmatrix}\Q\begin{bmatrix}\v\\ 1 \end{bmatrix} = \tr \left(\Q\begin{bmatrix}\v\\ 1 \end{bmatrix} \begin{bmatrix}\v^{T}& 1\end{bmatrix}\right)$. Denoting $\begin{bmatrix}\v\\ 1 \end{bmatrix} \begin{bmatrix}\v^{T}& 1\end{bmatrix}$ by $\H$ and dropping the rank one constraint, we obtain the following  \cite{Ben-Tal}  semi-definite relaxation of the optimization problem in \eqref{optgood}
\begin{eqnarray}
\rho = \max_{\H} \tr(\mathbf{QH})/4,\hspace{2mm}
\mathrm{s.t.}\hspace{1mm} \tr(\R\H)\leq 4,\label{sdrgood}\\
\H = \H^{T}\succeq 0,\hspace{1mm}
\H_{ii} = 1, i = 1,2,\ldots,N.
\end{eqnarray}
The relaxed solution of the SDR of equation (\ref{sdrgood}) can now be obtianed using numerical solvers (e.g. \cite{cvx}). Now, given the optimal solution $\widehat{\H}$ for the SDR, we propose below a randomization algorithm to obtain a solution ($\v$) for the unrelaxed problem (\ref{optgood}).   
\begin{algorithm}
\caption{Rounding algorithm}
\label{Algo1}
\begin{algorithmic}[1]
\State\textbf{Input:} $\widehat{\H}$, $K$.
\State\textbf{Initialize:} Set $\N = \phi$ and $\r = \sgn(\n)$, where $\n$ is a zero mean Gaussian random vector with covariance matrix $\widehat{\H}$. 
\State\textbf{Decision:} If $\r^{T}\R\r\leq4$ and $r_{N+1} = 1$, $\N = \N\bigcup\r$.
\State\textbf {Iterate:} from step-2 for $K>N$ steps.
\State\textbf{Output:} $\r^* = \argmax_{\r\in\N}\NrmTwo{\r}^2$, $\sigma = \NrmTwo{\r^*}^2$.
\end{algorithmic}
\end{algorithm}

Now, we obtain the bounds relating $\sigma$ (optimal value of the unrelaxed problem) and $\rho$ (optimal value of the SDR) using the proposed randomization algorithm.
\begin{thm}\label{Thm1}
Let $\Lambda$ denote the diagonal matrix with eigenvalues of $\R\widehat{\H}$. Let $p^2 = \pr((\R -\Lambda)\succeq\mathbf{0})$, then 
 \begin{equation}
\sigma\leq\rho\leq\frac{\pi}{2 }\left(\frac{1}{1-p}\right)\sigma.
\label{bndmain2}
\end{equation}
\end{thm}
\begin{proof}
The left side of the inequality holds because of the relaxation argument. To prove the right side, let $\widehat{\H}$ be the optimal solution of the SDR in \eqref{sdrgood}.
  So we have $\widehat{\H} = \widehat{\H}^{T}\succeq 0$, $\diag(\widehat{\H}) = \o$ and $\tr(\R\widehat{\H})\leq4$. We observe that $\widehat{\H}$ satisfies the properties of the covariance matrix of a Gaussian random vector.  Let $\n$ be a Gaussian random vector with mean $\mathbf{0}$ and covariance matrix $\widehat{\H}$. Let $\r = \sgn(\n)$, so that $r_i^2 = 1, i = 1,2,\ldots,N+1$. We can now choose a realization of $\r$ such that $r_{N+1} = 1$, so that it can be a feasible solution for the  optimization problem in  (\ref{eq:modified}). Clearly, such a realization of $\r$ always exists. Now, to obtain the right side of the inequality \eqref{bndmain2}, our goal is to show that such a realization $\r$ also satisfies $\theta\r^{T}\Q\r\geq \tr(\Q\widehat{\H})$ and $\r^{T}\R\r\leq4$ for some $\theta>1$.
  
Let $ f(\n)$ be the density of the Gaussian vector $\n$. We have
\begin{align*}
  &\pr\big(\r^{T}\Q\r \geq\frac{1}{\theta} \tr(\Q\widehat{\H}), \r^{T}\R\r\leq4,r_{N+1}=1\big)\\
  &= \int_{\r^{T}\Q\r \geq\frac{1}{\theta} \tr(\Q\widehat{\H}), \r^{T}\R\r\leq4,r_{N+1}=1} f(\n) \d \n\\
  &\stackrel{(a)}{=}  \int_{\r^{T}\Q\r \geq\frac{1}{\theta} \tr(\Q\widehat{\H}), \r^{T}\R\r\leq4,r_{N+1}=-1} f(\n) \d \n\\
  &=\pr\big(\r^{T}\Q\r \geq\frac{1}{\theta} \tr(\Q\widehat{\H}), \r^{T}\R\r\leq4,r_{N+1}=-1\big),
\end{align*}
where $\r = \sgn(\n)$. Here $(a)$ follows by the change of variables $\n \to -\n$ and the fact that a zero mean multivariate Gaussian density satisfies $f(\n)=f(-\n)$.
Hence
\begin{align*}
&\pr\big(\r^{T}\Q\r \geq\frac{1}{\theta} \tr(\Q\widehat{\H}), \r^{T}\R\r\leq4,r_{N+1}=1\big)\\
  &=\frac{1}{2}\pr\big(\r^{T}\Q\r \geq\frac{1}{\theta} \tr(\Q\widehat{\H}), \r^{T}\R\r\leq4\big).
\end{align*}
Hence it suffices to show that the RHS of the above equation is non zero. We have
\begin{flalign}
&\pr\big(\r^{T}\Q\r\geq\frac{1}{\theta}\tr(\Q\widehat{\H}),
\r^{T}\R\r\leq4\big),&\nonumber\\
\geq&\pr\big(\r^{T}\Q\r\geq\frac{1}{\theta}\tr(\Q\arcsin\widehat{\H}),
\r^{T}\R\r\leq4\big),&\label{a}
\end{flalign}
where (\ref{a}) follows from $\tr(\Q\widehat{\H})\leq\tr(\Q\arcsin\widehat{\H})$ (see \cite{Nestrov}). Now, applying the union bound to inequality (\ref{a}), we have
\begin{flalign}
\geq&\pr\big(\r^{T}\Q\r\geq\frac{1}{\theta}\tr(\Q\arcsin\widehat{\H})\big) -\pr\big(\r^{T}\R\r> 4\big),&\nonumber
\end{flalign}
Now, $\tr(\Q\arcsin\widehat{\H}) = \frac{\pi}{2}\mean\{\r^{T}\Q\r\}$ (see \cite{GWbook}). Using these in the above inequality, we have
\begin{flalign*}
\geq&\pr\big(\r^{T}\Q\r\geq\frac{\pi}{2\theta}\mean\{\r^{T}\Q\r\}\big) -\pr\big(\r^{T}\R\r>4\big),&
\end{flalign*}
Using the Paley-Zygmund inequality (see \cite{dasgupta2011probability}) in the first part
\begin{flalign}
\geq&\bigg(1-\bigg(\frac{\pi}{2\theta}\bigg)\bigg)^2\frac{\bigg(\mean\{\r^{T}\Q\r\}\bigg)^2}
{\mean\{(\r^{T}\Q\r)^2\}}-\pr\big(\r^{T}\R\r>4\big),&\nonumber\\
\geq&\bigg(1-\bigg(\frac{\pi}{2\theta}\bigg)\bigg)^2-\pr\big(\r^{T}\R\r>4 \big)&\label{b}
\end{flalign}
where we have used $\E[x^2] \geq (\E[x])^2$ in (\ref{b}). Let $\Lambda$ be a diagonal matrix of eigenvalues of $\R\widehat{\H}$. Since $\tr(\R\widehat{\H})\leq 4$, and $\r^T \Lambda \r = \tr(\R\widehat{\H})$ because $r_i^2=1$. Hence, (\ref{b}) reduces to
\begin{flalign}
&\geq\bigg(1-\bigg(\frac{\pi}{2\theta}\bigg)\bigg)^2-\pr\big(\r^{T}\R\r>\r^{T}\Lambda\r\big),&\\
&\geq\bigg(1-\bigg(\frac{\pi}{2\theta}\bigg)\bigg)^2-\pr\big(\R-\Lambda\succeq\mathbf{0}),&\\
&=\bigg(1-\bigg(\frac{\pi}{2\theta}\bigg)\bigg)^2-p^2,&
\end{flalign}
We get $\theta>\frac{\pi}{2}\big(\frac{1}{1-p}\big)$ by imposing $\big(1-\big(\frac{\pi}{2\theta}\big)\big)^2-p^2>0$.
\end{proof}

We observe that the necessary condition for the successful working of the randomization algorithm (step-$3$) is that the vectors $\r$ constructed in step-$2$ satisfy the quadratic constraint of the unrelaxed problem (\ref{optgood}). In the above proof, we have assumed that the probability of existence of such vectors $\r$ with non-zero probability, i.e. $\pr(\r^{T}\R\r\leq4)=1-p^2>0$. However, as $\r$ is a correlated Bernoulli random vector, evaluating the probability 
$\pr(\r^{T}\R\r\leq4)$ requires $2^{N}$ checks which cannot be performed in practice. Hence, we use a weaker upper bound on $p^2$ which depends on the known matrices $\R$ and $\widehat{\H}$ and which was observed (through simulations) to be satisfied by sparse networks (with node-density $\lambda\leq\frac{1}{\sqrt{N}}$). In bounding $p^2$, we have also taken care of appropriate probability measures. 

We would like to re-emphasize that the value of $\theta$ derived above in Theorem-\ref{Thm1} is not necessarily the best possible bound on $\rho$, because the upper bound on $p$ is very conservative since there is a major loss of probability measure while applying the union bound and also when we bound $\pr(\r^{T}\R\r>4)$ by $\pr(\R-\Lambda\succeq\mathbf{0})$. 
%---------------------------------------------------------------
% Simulation Start
%----------------------------------------------------------------
\section{Simulations}
%------------------------------------------------------------------
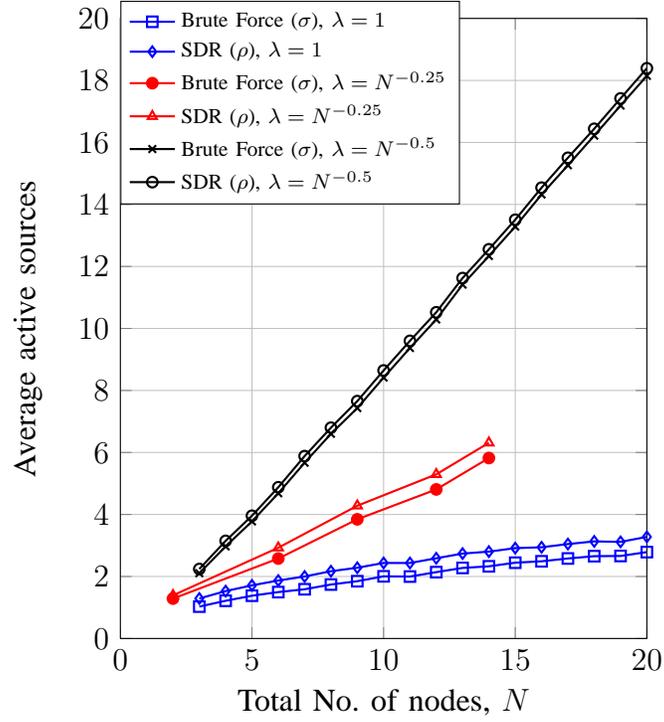
\begin{figure}
\centering
\begin{tikzpicture}
\draw[help lines] (0,0);
\begin{axis}[%
width=7cm,
height=8.25cm,
scale only axis,
xmin=0,
xmax=20,
grid=both,
major grid style={line width=.2pt,draw=gray!50},
xlabel={Total No. of nodes, $N$},
ymin=0,
ymax=20,
ylabel={Average active sources},
legend style={at={(0,0.7)},anchor=south west,draw=black,fill=white,legend cell align=left}
]

%\draw[step=1cm, black,thick] (-2,-2) grid (6,6);
%\draw[help lines] (0,0) grid (2,3);

\addplot [
color=blue,thick,
solid,
mark=square,
mark options={solid}
]
table[row sep=crcr]{
	3.0000    1.0290\\
    4.0000    1.2170\\
    5.0000    1.3800\\
    6.0000    1.4970\\
    7.0000    1.5880\\
    8.0000    1.7410\\
    9.0000    1.8480\\
   10.0000    2.0050\\
   11.0000    1.9980\\
   12.0000    2.1400\\
   13.0000    2.2760\\
   14.0000    2.3280\\
   15.0000    2.4450\\
   16.0000    2.4900\\
   17.0000    2.5810\\
   18.0000    2.6550\\
   19.0000    2.6620\\
   20.0000    2.7870\\
};
\addlegendentry{\scriptsize{Brute Force ($\sigma$), $\lambda = 1$}};

\addplot [
color=blue,thick,
solid,
mark=diamond,
mark options={solid}
]
table[row sep=crcr]{
	3.0000    1.2910\\
    4.0000    1.5289\\
    5.0000    1.7140\\
    6.0000    1.8699\\
    7.0000    1.9999\\
    8.0000    2.1702\\
    9.0000    2.2817\\
   10.0000    2.4360\\
   11.0000    2.4352\\
   12.0000    2.5881\\
   13.0000    2.7393\\
   14.0000    2.8022\\
   15.0000    2.9149\\
   16.0000    2.9425\\
   17.0000    3.0457\\
   18.0000    3.1325\\
   19.0000    3.1166\\
   20.0000    3.2780\\            
};
\addlegendentry{\scriptsize{SDR ($\rho$), $\lambda=1$}};
%----------------------------------------------------------------------
\addplot [
color=red,thick,
solid,
mark=*,
mark options={solid}
]
table[row sep=crcr]{
   	2.0000    1.2870\\
    6.0000    2.5790\\
    9.0000    3.8440\\
   12.0000    4.8080\\
   14.0000    5.8210\\        
};
\addlegendentry{\scriptsize{Brute Force ($\sigma$), $\lambda = N^{-0.25}$}};

\addplot [
color=red,thick,
solid,
mark=triangle,
mark options={solid}
]
table[row sep=crcr]{
   	2.0000    1.3916\\
    6.0000    2.9272\\
    9.0000    4.2818\\
   12.0000    5.2947\\
   14.0000    6.3120\\     
};
\addlegendentry{\scriptsize{SDR ($\rho$), $\lambda = N^{-0.25}$}};
%----------------------------------------------------------------------
\addplot [
color=black,thick,
solid,
mark=x,
mark options={solid}
]
table[row sep=crcr]{
	3.0000    2.1050\\
    4.0000    2.9710\\
    5.0000    3.7770\\
    6.0000    4.6840\\
    7.0000    5.6690\\
    8.0000    6.6020\\
    9.0000    7.4320\\
   10.0000    8.4160\\
   11.0000    9.3700\\
   12.0000   10.2900\\
   13.0000   11.4140\\
   14.0000   12.3370\\
   15.0000   13.2850\\
   16.0000   14.3200\\
   17.0000   15.2630\\
   18.0000   16.2200\\
   19.0000   17.1890\\
   20.0000   18.1470\\
};
\addlegendentry{\scriptsize{Brute Force ($\sigma$), $\lambda = N^{-0.5}$}};

\addplot [
color=black,thick,
solid,
mark=o,
mark options={solid}
]
table[row sep=crcr]{
	3.0000    2.2418\\
    4.0000    3.1483\\
    5.0000    3.9591\\
    6.0000    4.8773\\
    7.0000    5.8829\\
    8.0000    6.7991\\
    9.0000    7.6564\\
   10.0000    8.6451\\
   11.0000    9.6000\\
   12.0000   10.5211\\
   13.0000   11.6252\\
   14.0000   12.5525\\
   15.0000   13.5060\\
   16.0000   14.5412\\
   17.0000   15.5048\\
   18.0000   16.4399\\
   19.0000   17.4194\\
   20.0000   18.3918\\         
};
\addlegendentry{\scriptsize{SDR ($\rho$), $\lambda = N^{-0.5}$}};%----------------------------------------------------------------------

\end{axis}
\end{tikzpicture}%
\caption{$\rho$ versus $N$ for $\epsilon = 10$, path-loss exponent, $\beta=3$ and different network densities, $\lambda$.}
\label{fig:Fig01}
\end{figure}
%--------------------------------------
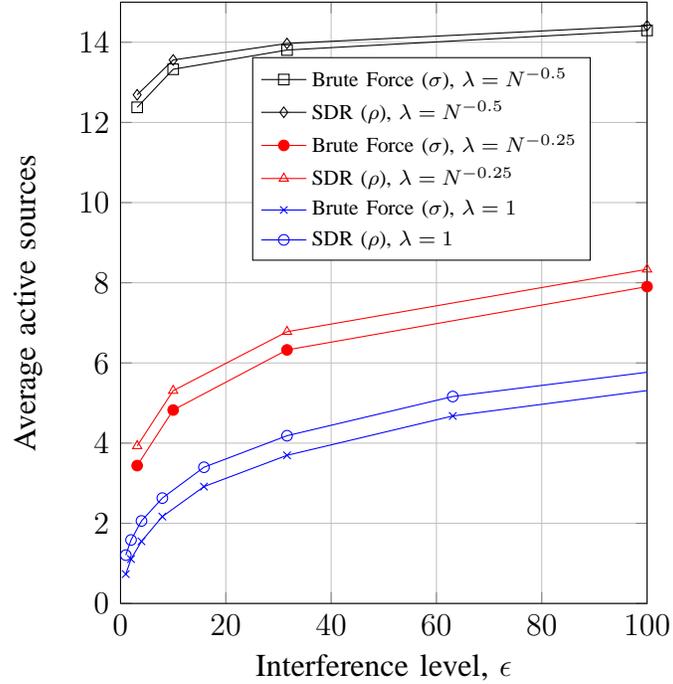
\begin{figure}
\centering
\begin{tikzpicture}
\draw[help lines] (0,0);
\begin{axis}[%
width=7cm,
height=8cm,
scale only axis,
xmin=0,
xmax=100,
grid=both,
grid style={line width=.1pt, draw=gray!10},
major grid style={line width=.2pt,draw=gray!50},
xlabel={Interference level, $\epsilon$},
ymin=0,
ymax=15,
ylabel={Average active sources},
legend style={at={(0.25,0.57)},anchor=south west,draw=black,fill=white,legend cell align=left}
]

%\draw[step=1cm, black,thick] (-2,-2) grid (6,6);
%\draw[help lines] (0,0) grid (2,3);

\addplot [
color=black,
solid,
mark=square,
mark options={solid}
]
table[row sep=crcr]{
  3.1623  12.3760\\
  10.0000 13.3260\\
  31.6228 13.8060\\
 100.0000 14.2950\\
% 316.2278 14.5630\\
% 1000     14.6610\\
% 3162.3   14.7670\\
% 10000    14.8660\\
% 31623    14.9080\\
% 100000   14.9580\\
};
\addlegendentry{\scriptsize{Brute Force ($\sigma$), $\lambda = N^{-0.5}$}};

\addplot [
color=black,
solid,
mark=diamond,
mark options={solid}
]
table[row sep=crcr]{
  3.1623  12.6889\\
  10.0000 13.5581\\
  31.6228 13.9696\\
 100.0000 14.4108\\
% 316.2278 14.6305\\
% 1000     14.7128\\
% 3162.3   14.7985\\
% 10000    14.8839\\
% 31623    14.9249\\
% 100000   14.9655\\        
};
\addlegendentry{\scriptsize{SDR ($\rho$), $\lambda = N^{-0.5}$}};
%----------------------------------------------------------------------
\addplot [
color=red,
solid,
mark=*,
mark options={solid}
]
table[row sep=crcr]{
  3.1623  3.4390\\
  10.0000 4.8250\\
  31.6228 6.3250\\
 100.0000 7.9040\\
};
\addlegendentry{\scriptsize{Brute Force ($\sigma$), $\lambda = N^{-0.25}$}};

\addplot [
color=red,
solid,
mark=triangle,
mark options={solid}
]
table[row sep=crcr]{
  3.1623  3.9295\\
  10.0000 5.3085\\
  31.6228 6.7788\\
 100.0000 8.3356\\     
};
\addlegendentry{\scriptsize{SDR ($\rho$), $\lambda = N^{-0.25}$}};
%----------------------------------------------------------------------
\addplot [
color=blue,
solid,
mark=x,
mark options={solid}
]
table[row sep=crcr]{
    1.0000    0.7320\\
    1.9953    1.1070\\
    3.9811    1.5500\\
    7.9433    2.1670\\
   15.8489    2.9150\\
   31.6228    3.6950\\
   63.0957    4.6770\\
  125.8925    5.7560\\
% 316.2278 14.5630\\
% 1000     14.6610\\
% 3162.3   14.7670\\
% 10000    14.8660\\
% 31623    14.9080\\
% 100000   14.9580\\
};
\addlegendentry{\scriptsize{Brute Force ($\sigma$), $\lambda =1$}};

\addplot [
color=blue,
solid,
mark=o,
mark options={solid}
]
table[row sep=crcr]{
  	1.0000    1.2060\\
    1.9953    1.5838\\
    3.9811    2.0585\\
    7.9433    2.6285\\
   15.8489    3.3981\\
   31.6228    4.1843\\
   63.0957    5.1625\\
  125.8925    6.1924\\
% 316.2278 14.6305\\
% 1000     14.7128\\
% 3162.3   14.7985\\
% 10000    14.8839\\
% 31623    14.9249\\
% 100000   14.9655\\        
};
\addlegendentry{\scriptsize{SDR ($\rho$), $\lambda = 1$}};
%----------------------------------------------------------------------
\end{axis}
\end{tikzpicture}%
\caption{$\rho$ versus $\epsilon$ for $N = 15$, path-loss exponent, $\beta=3$ and different network densities $\lambda$.}
\label{fig:Fig03}
\end{figure}
%------------------------------------------------------------
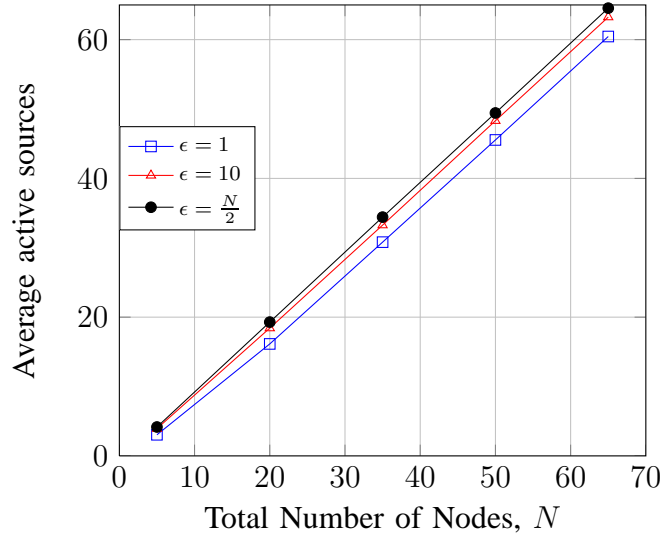
\begin{figure}
\centering
\begin{tikzpicture}
\draw[help lines] (0,0);
\begin{axis}[%
width=7cm,
height=6cm,
scale only axis,
xmin=0,
xmax=70,
grid=both,
grid style={line width=.1pt, draw=gray!10},
major grid style={line width=.2pt,draw=gray!50},
xlabel={Total Number of Nodes, $N$},
ymin=0,
ymax=65,
ylabel={Average active sources},
legend style={at={(0,0.5)},anchor=south west,draw=black,fill=white,legend cell align=left}
]

%\draw[step=1cm, black,thick] (-2,-2) grid (6,6);
%\draw[help lines] (0,0) grid (2,3);

\addplot [
color=blue,
solid,
mark=square,
mark options={solid}
]
table[row sep=crcr]{
    5.0000    3.0452\\
   20.0000   16.1403\\
   35.0000   30.8120\\
   50.0000   45.5406\\
   65.0000   60.4510\\
% 316.2278 14.5630\\
% 1000     14.6610\\
% 3162.3   14.7670\\
% 10000    14.8660\\
% 31623    14.9080\\
% 100000   14.9580\\
};
\addlegendentry{\scriptsize{$\epsilon=1$}};

\addplot [
color=red,
solid,
mark=triangle,
mark options={solid}
]
table[row sep=crcr]{
    5.0000    3.9387\\
   20.0000   18.4204\\
   35.0000   33.2660\\
   50.0000   48.3090\\
   65.0000   63.2483\\
% 316.2278 14.5630\\
% 1000     14.6610\\
% 3162.3   14.7670\\
% 10000    14.8660\\
% 31623    14.9080\\
% 100000   14.9580\\
};
\addlegendentry{\scriptsize{$\epsilon=10$}};

%\addplot [
%color=red,
%solid,
%mark=star,
%mark options={solid}
%]
%table[row sep=crcr]{
%    5.0000    4.1997\\
%   20.0000   18.7705\\
%   35.0000   33.6632\\
%   50.0000   48.7262\\
%   65.0000   63.7057\\
% 316.2278 14.5630\\
% 1000     14.6610\\
% 3162.3   14.7670\\
% 10000    14.8660\\
% 31623    14.9080\\
% 100000   14.9580\\
%};
%\addlegendentry{\scriptsize{$\epsilon=20$}};

\addplot [
color=black,
solid,
mark=*,
mark options={solid}
]
table[row sep=crcr]{
   5.0000    4.1677\\
   20.0000   19.2886\\
   35.0000   34.4212\\
   50.0000   49.4307\\
   65.0000   64.5405\\
% 316.2278 14.6305\\
% 1000     14.7128\\
% 3162.3   14.7985\\
% 10000    14.8839\\
% 31623    14.9249\\
% 100000   14.9655\\        
};
\addlegendentry{\scriptsize{$\epsilon=\frac{N}{2}$}};

\end{axis}
\end{tikzpicture}%
\caption{$\rho$ versus $N$ for different interference levels, $\epsilon$, path-loss exponent $\beta=3$ and for network density $\lambda = \frac{1}{\sqrt{N}}$.}
\label{fig:Fig05}
\end{figure}
%--------------------------------------------------------------------------
In this section, we present numerical simulations to check the performance of the proposed SDP relaxation and the randomization algorithm. In the simulation set-up, we generate $N=\lambda M$ points with a uniform distribution over a square $\el = [0,\sqrt{M}]^2$, where $\lambda$ is called the density of the network. We then evaluate the SDR given in (\ref{sdrgood}) using the MATLAB CVX toolbox \cite{cvx} to obtain the SDR solutions $\widehat{\H}$ and $\rho$. We compare the obtained SDR solution $\rho$, with the optimal packing $\sigma$, by solving (\ref{mainoptone}) using a brute force search for small values of $N\leq 20$. In the brute force search, $\sigma$ is obtained by iterating over all subset of nodes and picking the set with the largest cardinality that satisfies the constraints. We repeat the above experiment for $1000$ different spatial realizations for each $N$ with a fixed $\epsilon$ (and vice-versa) and average the values of $\rho$ and $\sigma$ obtained at each realization.

In Figure-\ref{fig:Fig01}, we compare the average $\rho$ (estimated by SDR) with average $\sigma$ (obtained by brute force search) as the number of nodes ($N$) increases for different network densities ($\lambda$) and a fixed interference level ($\epsilon = 10$). 

In Figure-\ref{fig:Fig03}, we compare the average $\rho$ (estimated by SDR) with average $\sigma$ (obtained by brute force search) as the interference level ($\epsilon$) increases for different network densities ($\lambda$) and a fixed number of nodes ($N = 15$). We observe that the number of active nodes initially increases sharply as the interference constraint is relaxed but stabilizes for large $\epsilon$. 

From Figures-\ref{fig:Fig01} and \ref{fig:Fig03}, we observe that $\sigma$ is always very close to $\rho$, in-fact we observe that $\rho-\sigma\leq 1$ in most cases. For the network density $\lambda = N^{-0.5}$, most of the nodes in the network are switched on. We also observed that the optimal vector, $\r$ obtained by the randomization algorithm was always same (after the transformation) as the optimal $\x$ obtained by brute force search for all realizations.   

In Figure-\ref{fig:Fig05}, we plot the average $\rho$ (estimated by SDR) as the number of nodes ($N$) increases for a fixed network density ($\lambda = \frac{1}{\sqrt{N}}$) and for different interference levels ($\epsilon$). We observe that changing the interference level has a small effect on the packing of network.

%We also performed simulations to check for the variation of $p^2 = \pr\big(\R-\Lambda\succeq\mathbf{0}\big)$ for different network densities (not shown here). We observed that the value of $p$ was less than $2/\pi$ for sparse networks ($\lambda\leq N^{-0.5}$) and was close to one for dense networks. Hence we can conclude that the upper bound on $p^2$ is tight for sparse networks.

From the simulations, we can conclude that the optimal packing of nodes depends critically on the density of the network. The other parameters such as the number of nodes ($N$), the interference level ($\epsilon$) and path-loss exponent ($\beta$) play a relatively minor role in the packing. 

We can also conclude that the optimal value of the packing problem obtained by the SDR is very close (In-fact, $\rho-\sigma<1$) to the actual optimal packing value. Hence, SDR is a tight approximation for the packing problem. 
%------------------------------------------------------------
%\input{fig07.tex}
%------------------------------------------------------------
% Conclusion Start
%----------------------------------------------------------------
\section{Conclusion}
In this work, we study the maximal packing problem under interference constraints where the goal is to find the maximum number of active nodes in an area such that the total interference in the network is less than some fixed value $\epsilon$. This is a discrete optimization problem which is NP hard. We propose a semi-definite relaxation (SDR) of the NP hard problem, whose solution upper bounds the number of active nodes in the network. Simulations are performed to compare the bounds provided by the SDR with a brute-force search solution of the unrelaxed problem for small networks and we observe that the SDP relaxation provides a good approximation to the packing problem. 
%-------------------------------
% Concl End
%----------------------------------------------------------------------
% Bib
%----------------------------------------------------------------------
\bibliographystyle{IEEEtran}
\bibliography{IEEEabrv,Bib_One}
%-----------------------------
\end{document}